\documentclass[copyright,creativecommons]{eptcs}
\providecommand{\event}{ACT 2023} 

\usepackage{iftex}
\usepackage{amsfonts,amsmath,amssymb,amsthm}
\usepackage[english]{babel}
\usepackage{bm,bbm}
\usepackage{graphicx}
\usepackage{pifont}
\usepackage{semantic}
\usepackage{stmaryrd}
\usepackage{subcaption}

\usepackage{satex}
\deftwocell[circle]{copy: 1 -> 2}
\deftwocell[braidl]{swap: 2 -> 2}
\deftwocell[circle,"\cdot"]{del: 1 -> 0}

\deftwocell[circle,"f"]{f: 1 -> 1}
\deftwocell[circle,"g"]{g: 1 -> 1}
\deftwocell[circle,"h_1"]{h1: 1 -> 1}
\deftwocell[circle,"h_2"]{h2: 1 -> 1}
\defsatexfig{causalpremisel}{
    \twocell{
        f *
        g *
        copy *
        (h1 * 1)
    }
}
\defsatexfig{causalpremiser}{
    \twocell{
        f *
        g *
        copy *
        (h2 * 1)
    }
}
\defsatexfig{causalconclusionl}{
    \twocell{
        f *
        copy *
        (g * 1) *
        (copy * 1) *
        (h1 * 2)
    }
}
\defsatexfig{causalconclusionr}{
    \twocell{
        f *
        copy *
        (g * 1) *
        (copy * 1) *
        (h2 * 2)
    }
}

\deftwocell[circle,"p"]{prior: 0 -> 1}
\deftwocell[rectangle,"k"]{k: 2 -> 1}
\defsatexfig{pushback}{
    \twocell{
        (prior * label[u,"Z"]) *
        ((1 -> 1)[blank,"\Upsilon"] * 1) *
        k *
        label[d,"X"]
    }
}

\deftwocell[circle,"U"]{uniform: 0 -> 1}
\defsatexfig{borelpushback}{
    \twocell{
        (uniform * label[u,"Z"]) *
        ((1 -> 1)[blank,"\Omega"] * 1) *
        k *
        label[d,"X"]
    }
}

\deftwocell[circle,"f"]{conditional: 1 -> 1}
\defsatexfig{jointCoPara}{
    \twocell{
        label[u,"Z"] *
        copy *
        (conditional * 1) *
        (copy * 1) *
        (1 * k) *
        (label[d,"M"] * label[d,"X"])
    }
}

\defsatexfig{jointMarkov}{
    \twocell{
        label[u,"Z"] *
        copy *
        (conditional * 1) *
        k *
        label[d,"X"]
    }
}

\deftwocell[circle,"f_1"]{f1: 1 -> 1}
\deftwocell[circle,"f_2"]{f2: 1 -> 1}
\deftwocell[rectangle,"k_1"]{k1: 2 -> 1}
\defsatexfig{jointCompStoch}{
    \twocell{
        label[u,"Z"] *
        copy *
        (f1 * 1) *
        (copy * 1) *
        (1 * k1) *
        (1 * f2) *
        (label[d,"M_1"] * label[d,"M_2"])
    }
}

\deftwocell[rectangle,"k_2"]{k2: 2 -> 1}
\defsatexfig{jointCompDet}{
    \twocell{
        (label[u,"M_1"] * label[u,"M_2"] * label[u,"Z"]) *
        (swap * 1) *
        (1 * k1) *
        k2 *
        label[d,"Y"]
    }
}

\defsatexfig{jointProdStoch}{
    \twocell{
        (label[u,"Z"] * label[u,"W"]) *
        (f1 * f2) *
        (label[d,"M_1"] * label[d,"M_2"])
    }
}

\defsatexfig{jointProdDet}{
    \twocell{
        (label[u,"M_1"] * label[u,"M_2"] * label[u,"Z"] * label[u,"W"]) *
        (1 * swap * 1) *
        (k1 * k2) *
        (label[d,"X"] * label[d,"Y"])
    }
}

\deftwocell[rectangle,"f"]{pf: 1 -> 1}
\deftwocell[rectangle,"g"]{pg: 1 -> 1}
\defsatexfig{densPairing}{
    \twocell{
        label[u,"Z"] *
        copy *
        (pf * pg) *
        (label[d,"X"] * label[d,"Y"])
    }
}

\deftwocell[rectangle,"u"]{pu: 1 -> 2}
\defsatexfig{uniformCircleGraph}{
    \twocell{
        label[u,"\mathbb{R}^{2\times 2}"] *
        pu *
        (label[d,"\mathbb{R}"] * label[d,"\mathbb{R}"])
    }
}

\deftwocell[circle,"U"]{u: 0 -> 1}
\deftwocell[rectangle,"k"]{ku: 2 -> 2}
\defsatexfig{uniformCircleDiagram}{
    \twocell{
        label[u,"\mathbb{R}^{2\times 2}"] *
        (u * 1) *
        (copy * 1) *
        (1 * ku) *
        (label[d,"[0, 2\pi]"] * label[d,"\mathbb{R}"] * label[d,"\mathbb{R}"])
    }
}

\deftwocell[rectangle,"p"]{fake: 0 -> 1}
\deftwocell[rectangle,"\Delta"]{diff: 1 -> 1}
\defsatexfig{fakeCoinGraph}{
    \twocell{
        fake *
        diff *
        label[d,"\mathbb{R}"]
    }
}

\deftwocell[circle,"p"]{pfake: 0 -> 1}
\deftwocell[rectangle,"\mathbb{I}"]{indic: 1 -> 2}
\deftwocell[circle,"\Delta \oplus \delta_{0}",labelwidth=2]{pdiff: 1 -> 1}
\deftwocell[rectangle,"\pi"]{project: 1 -> 1}
\defsatexfig{fakeCoinDiagram}{
    \twocell{
        pfake *
        copy *
        (1 * pdiff) *
        (1 * project) *
        (label[d,"\mathbbm{2}"] * label[d,"\mathbb{R}"])
    }
}

\ifpdf
  \usepackage{underscore}         
  \usepackage[T1]{fontenc}        
\else
  \usepackage{breakurl}           
\fi

\newcommand{\dom}[1]{\mathrm{dom}(#1)}
\newcommand{\cod}[1]{\mathrm{cod}(#1)}
\newcommand{\dirac}[2]{\delta_{#1}(#2)}
\newcommand{\expectint}[3]{\int_{#2}\: #3 \: #1}

\newcommand{\weight}{\mathbb{R}_{\geq 0}}
\newcommand{\extWeight}{[0, \infty]}

\newcommand{\mswap}[2]{\mathbf{swap}_{#1 \otimes #2}}
\newcommand{\mcopy}[1]{\mathbf{copy}_{#1}}
\newcommand{\mdel}[1]{\mathbf{del}_{#1}}
\newcommand{\Measure}[1]{\mathbb{M}\left(#1\right)}
\newcommand{\prior}[3]{p_{\scriptsize #1}(#2 \mid #3)}
\newcommand{\Prob}[1]{\mathbb{P}\left(#1\right)}

\newcommand{\C}{\mathcal{C}}
\newcommand{\BorelStoch}{\mathbf{BorelStoch}}
\newcommand{\CoPara}[1]{\mathbf{CoPara}_{\otimes}(#1)}
\newcommand{\CoParaB}[1]{\mathbf{CoPara}_{\bullet}(#1)}
\newcommand{\Dens}{\mathbf{Dens}}
\newcommand{\FinHyp}{\mathbf{FinHyp}}
\newcommand{\FreeCD}[1]{\mathbf{FreeCD}_{#1}}
\newcommand{\FreeMarkov}[1]{\mathbf{FreeMarkov}_{#1}}
\newcommand{\dJoint}{\partial\mathbf{Joint}}
\newcommand{\Joint}{\mathbf{Joint}}
\newcommand{\Hyp}{\mathbf{Hyp}}
\newcommand{\hyp}[1]{\mathrm{hyp}(#1)}

\newcommand{\M}{\mathcal{M}}
\newcommand{\Meas}{\mathbf{Meas}}

\newcommand{\Para}[1]{\mathbf{Para}_{\otimes}(#1)}
\newcommand{\ParaB}[1]{\mathbf{Para}_{\bullet}(#1)}

\newcommand{\QBS}{\mathbf{QBS}}
\newcommand{\Sbs}{\mathbf{Sbs}}
\newcommand{\sem}[1]{\left\llbracket #1 \right\rrbracket}
\newcommand{\sfKrn}{\mathbf{sfKrn}}

\newcommand{\sfStoch}{\mathbf{sfStoch}}
\newcommand{\Stoch}{\mathbf{Stoch}}

\newtheorem{definition}{Definition}
\newtheorem{example}{Example}

\newtheorem{proposition}{Proposition}
\newtheorem{theorem}{Theorem}
\newtheorem{corollary}[theorem]{Corollary}

\title{String Diagrams with Factorized Densities}
\author{Eli Sennesh
\institute{Khoury College of Computer Science\\
Northeastern University\\
Boston, Massachusetts, United States of America}
\email{sennesh.e@northeastern.edu}
\and
Jan-Willem van de Meent
\institute{Amsterdam Machine Learning Lab (AMLab)\\
University of Amsterdam\\
Amsterdam, the Netherlands}
\email{\quad j.w.vandemeent@uva.nl}
}
\def\titlerunning{String Diagrams with Factorized Densities}
\def\authorrunning{E. Sennesh \& J-W. van de Meent}
\begin{document}
\maketitle

\begin{abstract}
A growing body of research on probabilistic programs and causal models has highlighted the need to reason compositionally about model classes that extend directed graphical models. Both probabilistic programs and causal models define a joint probability density over a set of random variables, and exhibit sparse structure that can be used to reason about causation and conditional independence. This work builds on recent work on Markov categories of probabilistic mappings to define a category whose morphisms combine a joint density, factorized over each sample space, with a deterministic mapping from samples to return values. This is a step towards closing the gap between recent category-theoretic descriptions of probability measures, and the operational definitions of factorized densities that are commonly employed in probabilistic programming and causal inference.
\end{abstract}

\section{Introduction}
\label{sec:intro}

Statisticians and machine learners analyze observed data by synthesizing models of those data. These models take a variety of forms, with several of the most widely used being directed graphical models, probabilistic programs, and structural causal models (SCMs). Applications of these frameworks have included cognitive modeling~\cite{Chater2006,Lake2017}, simulation-based inference~\cite{Cranmer2020}, and model-based planning~\cite{Friston2017,Levine2018}. Unfortunately, the richer the model class, the weaker the mathematical tools available to reason rigorously about it: SCMs built on linear equations with Gaussian noise admit easy inference, while graphical models have a clear meaning and a wide array of inference algorithms but encode a limited family of models. Probabilistic programs can encode any computably sampleable distribution, but the definition of their densities commonly relies on operational analogies with directed graphical models.

In recent years, category theorists have developed increasingly sophisticated ways to reason diagrammatically about a variety of complex systems. These include (co)parameterized categories of systems that may modify their parameters~\cite{Capucci2021} and hierarchical string diagrams for rewriting higher-order computations~\cite{AlvarezPicallo2022}. Recent work on Markov categories of probabilistic mappings has provided denotational semantics to probabilistic programs \cite{Staton2017,Heunen2017}, abstract categorical descriptions of conditioning, disintegration, sufficient statistics, conditional independence \cite{ChoJacobs2019,Fritz2020}, and generalized causal models~\cite{FritzKlingler2023,FritzLiang2023}.

This paper will take a step towards closing the gap between categorical probability and operational practice in probabilistic programming and applied Bayesian statistics. Denotational semantics for probabilistic programs define a measure over return values of a program given its inputs \cite{Staton2017,Heunen2017}. To reason about inference methods, practitioners need to consider the joint distribution of internal random variables, as well as its density's factorization into conditionals. Section~\ref{sec:background} will review basic definitions from probability and measure theory necessary to do so. Section~\ref{sec:gadgets} will then develop a category whose morphisms express joint (rather than marginal) distributions with factorized joint densities. Section~\ref{sec:diagram_scm} will show that generalized causal models can factorize these densities and admit interventional and counterfactual queries. Section~\ref{sec:discussion} will work through a pair of examples and summarize the paper's developments.


Appendix~\ref{sec:measure_theory} reviews the measure-theoretic concepts employed here; Appendix~\ref{sec:para_copara} reviews parametric and coparametric categories~\cite{Capucci2021}; and Appendix~\ref{sec:free_constructions} reviews free copy/delete and Markov categories.

\paragraph{Notation} The notation $(\C, \otimes, I)$ will range over strict symmetric monoidal categories (SMC's for short). We denote composition as $g \circ f$ or equivalently as $f\fatsemi g$, write $(X^{*}, \odot, ())$ for the finite list monoid on $X$'s, and overload $\otimes$ and $\oplus$ for direct products and sums. We draw string diagrams from the top (domain) to the bottom (codomain), showing products from left to right. Given a Markov category $\C$ we will draw deterministic maps in $\C_{det} \subset \C$ (which commute with $\mcopy{}$) as rectangles and stochastic ones as ellipses/circles. We nest brackets with parentheses $([])$ equivalently.

\section{Background: abstract and concrete categorical probability}
\label{sec:background}
This section will review the background on which the rest of the paper builds. Categorical probability begins from an abstract notion of nondeterminism: processes with a notion of ``independent copies''. Categorical probability then refines from a setting in which those nondeterministic processes ``happen'' whether observed or not, to a refined setting in which processes only ``happen'' when they affect an observed output. Categories of probability kernels, taking into account the details of measure theory (see Appendix~\ref{sec:measure_theory}), will form a concrete instance of the abstract setting.




Definition~\ref{def:cdcat} represents nondeterministic processes abstractly. A copy/delete category is an SMC whose morphisms generate information which can be copied or deleted freely.
\begin{definition}[Copy/delete category]
\label{def:cdcat}
A copy-delete or \emph{CD-category} is an SMC $(\C, \otimes, I)$ in which every object $X \in Ob(\C)$ has a commutative comonoid structure $\mcopy{X}: \C(X, X \otimes X)$ and $\mdel{X}: \C(X, I)$ which commutes with the monoidal product structure.
\end{definition}
Definition~\ref{def:markov_cat} then refines the abstract setting of CD categories to require that deleting the only result of a nondeterministic process is equivalent to deleting the process itself.
\begin{definition}[Markov category]
\label{def:markov_cat}
A \emph{Markov category} is a semicartesian CD-category $(\C, \otimes, I)$, so that the comonoidal counit is natural ($\forall f: \C(Z, X), f \fatsemi \mdel{X} = I$) and makes $I \in Ob(\C)$ a terminal object.
\end{definition}
Example~\ref{ex:meas} gives the canonical Markov category, consisting of measurable spaces and maps.
\begin{example}[Measurable spaces and functions form a category~\cite{Tao2011}]
\label{ex:meas}
Measurable spaces and functions form a Cartesian category $\Meas$ with objects $(X, \Sigma_X) \in Ob(\Meas)$ consisting of sets $X \in Ob(\mathbf{Set})$ and their $\sigma$-algebras\footnote{Collections of ``measurable subsets'' closed under complements, countable unions, and countable intersections} $\Sigma_{X}$ and morphisms $\Meas((Z, \Sigma_Z), (X, \Sigma_X)) = \left\{ f \in X^{Z} \mid \forall \sigma_X \in \Sigma_X, f^{-1}(\sigma_X) \in \Sigma_Z \right\}$ consisting of measurable functions between measurable spaces.
\end{example}
$\Meas$ acquires its Markov comonoid structure from its Cartesian structure.
Definition~\ref{def:stoch} below provides the canonical Markov category for measure-theoretic probability.
\begin{definition}[Category of measurable spaces and Markov kernels]
\label{def:stoch}
The category $\Stoch = Kl(\mathbb{P})(\Meas)$ of measurable spaces and Markov kernels is the Kleisli category of the Giry monad~\cite{Giry1982} over $\Meas$, having measurable spaces as objects and Markov kernels (Definition~\ref{def:markov_kernel}) between them as morphisms.
\end{definition}
Much of this paper will require a \emph{strict} Markov category as in Definition~\ref{def:strict_markov} below.
\begin{definition}[Strict Markov category]
\label{def:strict_markov}
A \emph{strict} Markov category is one whose underlying SMC (with comonoid structure thrown away) is strict monoidal (its associator and unitors are identity).
\end{definition}
Theorem 10.17 in Fritz~\cite{Fritz2020} showed that every Markov category is comonoid equivalent to a strict one, licensing us to work with strictified Markov categories $\Meas$ and $\Stoch$ without further concern.



Unless otherwise mentioned, this paper will work with $\Meas$ and $\Stoch$ as strict, causal Markov categories\footnote{The latter property is shown in Example 11.35 of Fritz~\cite{Fritz2020}}. When the ambient category and $\sigma$-algebra is clear from context, $f: Z \leadsto X$ will abbreviate $f: \Stoch((Z, \Sigma_Z), (X, \Sigma_X))$. In the concrete case of $\Stoch$, measurable maps give the deterministic maps $\Meas \simeq \Stoch_{det} \subseteq \Stoch$. While Markov categories provide a compositional setting for nondeterministic processes, Markov kernels in these categories only provide probability measures for their outputs given their inputs. By design, they ``forget'' (i.e.~marginalize over) all intermediate randomness in long chains of composition. Section~\ref{sec:gadgets} will build up a novel setting that ``remembers'' (i.e.~does not marginalize over) joint distributions over all intermediate random variables through long chains of composition, and will show when there exist probability densities with respect to the joint distributions thus formed.

\section{Joint distributions and densities for string diagrams}
\label{sec:gadgets}

Statisticians cannot utilize input-output (parameter to distribution) mappings alone, except for maximum likelihood estimation. Instead, these typically appear as conditional probability distributions in a larger probability model. This larger model necessarily encodes a \emph{joint distribution} over all relevant random variables, both those observed as data and the \emph{latent} variables that give rise to observations. Practical probabilistic reasoning then consists of applying the laws of probability (product law for conjunctions, sum law for disjunctions, marginalization for unconditional events, Bayesian inversion) to numerical \emph{densities} representing the joint distribution. This section will model the algebra of joint probability densities in a novel Markov category $\Joint$ constructed on the underlying Markov category $\Stoch$.

Section~\ref{subsec:joint} will first review an abstraction for categories in which morphisms act by ``pushing forward'' an internal ``parameter space'' and then instantiate that abstraction on a Markov category to yield a Markov category $\Joint$ of joint distributions. Section~\ref{subsec:sbs_densities} will give the conditions for a concrete Markov kernel to admit a density. Section~\ref{subsec:joint_densities} will use those preliminaries to define a Markov category whose morphisms generate and push forward a joint probability density.


\subsection{Accumulating random variables into joint distributions}
\label{subsec:joint}

Structural graphical models and probabilistic programs separate between the functions and variables they allow into deterministic and random ones~\cite{Pearl2012}. Representing deterministic mechanisms categorically requires assuming that each nondeterministic process consists of a deterministic mechanism and a (potentially conditional) distribution over a random variable. This subsection will exploit ``cybernetic'' constructions (overviewed in Appendix~\ref{sec:para_copara}) for parameterization of deterministic mechanisms by random inputs and ``writing out'' of internal joint distributions as coparameters.

Proposition~\ref{prop:stoch_m_actegory} will show the concrete category $\Stoch$ supports those constructions.
\begin{proposition}[$\Stoch$ forms a symmetric monoidal $\M$-actegory\footnote{Definition~\ref{def:M_actegory} in Appendix~\ref{sec:para_copara}}]
\label{prop:stoch_m_actegory}
The concrete category $\Stoch$ forms a symmetric monoidal $\M$-actegory for $\M = \Stoch$ and $\C = \Stoch$.
\end{proposition}
\begin{proof}
Any SMC $\C$ forms a symmetric monoidal $\M$-actegory for $\M = \C$ with the product functor $\M \bullet \C = \C \times \C$ from the product category. Any Markov category is also an SMC, and so $\Stoch$ suffices.
\end{proof}

In this trivial case, Definition~\ref{def:para_full} simplifies so that Definition~\ref{def:para} will form an SMC.
\begin{definition}[Symmetric monoidal parametric categories]
\label{def:para}
Given a strict SMC \((\C, \otimes, I)\), the symmetric monoidal \emph{parametric (bi)category} \(\Para{\C}\) has as objects those of $\C$ and as morphisms the pairs \(\Para{\C}(A, B) = \left\{ (M, k) \in Ob(\C) \times \C(M \otimes A, B) \right\}\). Composition for the two parameterized morphisms \((M, k): \Para{\C}(A, B)\) and \((M', k'): \Para{\C}(B, C)\) consists of \((M' \otimes M, k' \circ (id_{M'} \otimes k)))\); identities on objects \(A\) consist of \((I, id_{A})\); and $(\Para{\C}, \otimes, I)$ inherits its monoidal structure from $\C$\footnote{see Proposition~\ref{prop:para_copara_monoidal}}.
\end{definition}
$\Para{\Stoch}$ will suffice for Definition~\ref{def:joint_kernel} to model a Markov kernel over a joint distribution. The jointly random \emph{residual} $(M, \Sigma_M) \in Ob(\Stoch)$ will parameterize the deterministic map $k$.
\begin{definition}[Joint Markov kernel]
\label{def:joint_kernel}
A \emph{joint Markov kernel} is a pair of a Markov kernel with a deterministic mapping parameterized by that Markov kernel, up to permutation of residual components
\begin{align*}
    \Joint(Z, X) &:= \left\{ (f, [M, k]): \Stoch(Z, M) \times \Para{\Stoch_{det}}(Z, X) \right\}.
\end{align*}
\end{definition}
As implied by the hom-set notation above, joint Markov kernels will form a category of nondeterministic processes. Since the residuals of joint distributions only contribute to downstream processes through their local outputs, Theorem~\ref{thm:joint} will show this to be a copy/delete category.

\begin{theorem}[Joint Markov kernels form a copy/delete category]
\label{thm:joint}
$\Joint$ is a strict copy/delete category having $Ob(\Joint) = Ob(\Stoch)$ and joint Markov kernels as morphisms.
\end{theorem}
\begin{proof}
$\Joint$ must admit the typical requirements of a category as well as deterministic, copy-delete symmetric monoidal structure. We can demonstrate the necessary deterministic structure by exhibiting joint kernels $(I, k): \Para{\Stoch_{det}}(Z, X) \implies ([I, k], \mdel{Z}): \Joint(Z, X)$ for any noiseless causal mechanism.
Setting $k=\mcopy{X}$ or $k=\mdel{Z}$ yields the necessary copy and delete maps. Setting $k=\mswap{Z}{X}$ gives the necessary symmetry of the monoidal product. It remains to show that $\Joint$ has a monoidal product over morphisms and that its hom-sets are closed under composition.

Given two joint Markov kernels $(f_1, [M_1, k_1]): \Joint(Z, X)$ and $(f_2, [M_2, k_2]): \Joint(W, Y)$, their monoidal product is formed by pairing their causal mechanisms and noise distributions
\begin{align*}
    (f_1, [M_1, k_1]) \otimes (f_2, [M_2, k_2]) &:= (f_1 \otimes_{\Stoch} f_2, [M_1, k_1] \otimes_{\Para{\Stoch_{det}}} [M_2, k_2]): \Joint(Z \otimes W, X \otimes Y).
\end{align*}
Composing two joint Markov kernels $(f_1, [M_1, k_1]): \Joint(Z, X)$ and $(f_2, [M_2, k_2]): \Joint(X, Y)$ along their intermediate object involves composing their parametric maps and taking a conditional product of their stochastic kernels to form the composite joint distribution
\begin{align}
    \label{eq:joint_comp}
    (f_1, [M_1, k_1]) \fatsemi (f_2, [M_2, k_2]) &:= \left(\usesatexfig{jointCompStoch}, \left(M_1 \otimes M_2, \usesatexfig{jointCompDet} \right) \right): \Joint(Z, Y).
\end{align}
\end{proof}

Anything called a \emph{joint} Markov kernel ought to expose its internal joint distribution in a structured way. Definition~\ref{def:copara} will link the composition of joint distributions to the cybernetics literature.
\begin{definition}[Symmetric monoidal coparametric categories\footnote{See Definition~\ref{def:copara_full} for the more general case}]
\label{def:copara}
Given a strict SMC \((\C, \otimes, I)\), the symmetric monoidal \emph{coparametric (bi)category} \(\CoPara{\C}\) has as objects those of $\C$ and as morphisms the pairs \(\CoPara{\C}(A, B) = \left\{ (M, k) \in Ob(\C) \times \C(A, M \otimes B) \right\}\) of a residual object and a morphism from $A$ to $M \otimes B$. Composition for the morphisms \((M, k): \CoPara{\C}(A, B)\) and \((M', k'): \CoPara{\C}(B, C)\) consists of \((M' \otimes M, (id_{M} \otimes k') \circ k))\); identities on objects \(A\) consist of \((I, id_{A})\); and $(\CoPara{\C}, \otimes, I)$ inherits its monoidal structure from $\C$\footnote{See Proposition~\ref{prop:para_copara_monoidal}}.
\end{definition}
$\Joint$ serves to work with joint distributions compositionally rather than marginalizing them out. Theorem~\ref{thm:joint_copara} will show how mapping from $\Joint \rightarrow \CoPara{\Stoch}$ exposes the full joint distribution.
\begin{theorem}[Joint Markov kernels coparameterize joint distributions]
\label{thm:joint_copara}
There exists a full, identity-on-objects Markov functor $\sem{\cdot}: \Joint \rightarrow \CoPara{\Stoch}$ which maps the residual of a joint Markov kernel in $\Joint$ onto the residual of its image in $\CoPara{\Stoch}$.
\end{theorem}
\begin{proof}
The required functor sends morphisms $\sem{\cdot}: \Joint(Z, X) \rightarrow \CoPara{\Stoch}(Z, X)$ to coparameterized Markov kernels whose codomain is the joint distribution over the residual and the output
\begin{align*}
    \sem{(f, [M, k])} &= \left(M, \usesatexfig{jointCoPara}\right).
\end{align*}
This functor is trivially full, since any morphism $f: \CoPara{\Stoch}(Z, X)$ embeds trivially into $\Joint$ by setting the corresponding deterministic $k=id_{M \otimes X}$. It is not faithful: multiple ``divisions of labor'' between $f$ and $k$ can yield the same Markov kernel in $\CoPara{\Stoch}$.
\end{proof}

Corollary~\ref{cor:joint_markov} will give the trivial extension of marginalizing over the residual.
\begin{corollary}[Marginalizing a joint Markov kernel's residual yields a Markov kernel]
\label{cor:joint_markov}
There exists a full, identity-on-objects functor $J: \Joint \rightarrow \Stoch$.
\end{corollary}
\begin{proof}
The required functor $J$ just applies $\sem{\cdot}$ and then forgets the residual by composition with $\mdel{M}$: its action on morphisms is $J((f, [M, k])) = \sem{(f, [M, k])} \fatsemi (\mdel{M} \otimes id_{X})$.
\end{proof}

This subsection has considered arbitrary, unstructured joint distributions $\Joint$. Section~\ref{subsec:sbs_densities} will examine the special case in which the residual object is a standard Borel space and the conditional distribution into it meets the necessary conditions to admit a probability density.

\subsection{Base measures and densities over standard Borel spaces}
\label{subsec:sbs_densities}
Applied probability typically works not with probability measures but with probability densities, functions over a finite-dimensional sample space giving the ``derivative'' of a probability measure at a point. However, probability densities only exist for measures that meet the conditions of the Radon-Nikodym Theorem, and only relative to a specified base measure over the sample space. This section will restrict the residual objects or internal noises of joint Markov kernels to standard Borel sample spaces admitting probability densities, and then show that this restriction still admits a broad class of joint Markov kernels.

Definition~\ref{def:measure_space_cat} provides a suitable ambient category for base measures.
\begin{definition}[Category of measure spaces]
\label{def:measure_space_cat}
The \emph{category of measure spaces} $\mathbb{M}$ has as objects the measure spaces $(X, \Sigma_X, \mu)$ (Definition~\ref{def:measure_space}) and as morphisms the measure-preserving maps
\begin{align*}
    \mathbb{M}\left((Z, \Sigma_Z, \mu_Z), (X, \Sigma_X, \mu_X)\right) &= \left\{ f: \Meas((Z, \Sigma_Z), (X, \Sigma_X)) \mid \forall \sigma_X \in \Sigma_X, \mu_Z(f^{-1}(\sigma_X)) = \mu_X(\sigma_X) \right\}.
\end{align*}
\end{definition}
Applications typically deal with probability densities over finite-dimensional Euclidean spaces and countable sets. In $\Meas$, these can be characterized by the standard Borel spaces $\Sbs \subset \Meas$, which are unique for each cardinality up to uncountability. Assigning these their canonical base measures will provide a suitable setting of measure spaces for characterizing densities.

However, the Radon-Nikodym Theorem requires that the sample space admit not only a measure but a $\sigma$-finite (Definition~\ref{def:sigma_finite}) base measure. Proposition~\ref{prop:measure_space_product} and Proposition~\ref{prop:measure_space_sum} will therefore characterize the algebraic operations under which $\sigma$-finite measure spaces are closed. Proposition~\ref{prop:measure_space_product} below will characterize the base measures for joint probability densities.
\begin{proposition}[$\sigma$-finite measure spaces have finite direct products]
\label{prop:measure_space_product}
Let $I \in Ob(\mathbf{FinSet})$ be a set and let there be an $I$-indexed family of $\sigma$-finite measure spaces $(X_i, \Sigma_{X_i}, \mu_{X_i})_{i \in I} \in Ob(\mathbb{M})$. Then there exists a $\sigma$-finite \emph{direct product} measure space $\bigotimes_{i \in I} (X_i, \Sigma_{X_i}, \mu_{X_i}) = (X, \Sigma_X, \mu_X)$.
\end{proposition}
\begin{proof}
The product $\bigotimes_{i \in I} (X_i, \Sigma_i) \in Ob(\Meas)$ exists thanks to $\Meas$ being Cartesian, so that the resulting set is that of Cartesian products and the $\sigma$-algebra is also that of Cartesian products. Letting $\pi_i$ be the projection indexed by $i \in I$ of a Cartesian product, we write the $\sigma$-finite product measure (which exists and is unique when $(X_i, \Sigma_{X_i}, \mu_{X_i})$ are $\sigma$-finite~\cite{Tao2011}\footnote{Definition 1.7.4, page 161}) as $\mu_X(\sigma_{X}) = \prod_{i \in I} \mu_{X_i}\left(\left\{ \pi_i(x): x \in \sigma_{X} \right\}\right)$, yielding the direct product $\left( \bigotimes_{i \in I} X_i, \bigotimes_{i \in I} \Sigma_{X_i}, \mu_X \right) \in Ob(\mathbb{M})$.
\end{proof}
The reader can check that the direct product of measure spaces does not form a categorical product: the pairing required to witness the universal property will not be measure-preserving, with intervals of different lengths in the real line providing a counterexample.

Proposition~\ref{prop:measure_space_sum} will then characterize the base measures for mixture probability densities.
\begin{proposition}[$\sigma$-finite measure spaces have countable direct sums~\cite{Fremlin2010b}\footnote{214L, page 38}]
\label{prop:measure_space_sum}
Let $I \in Ob(\mathbf{Set})$ be a countable set and $(X_i, \Sigma_{X_i}, \mu_{X_i})_{i \in I} \in Ob(\mathbb{M})$ be a family of $\sigma$-finite measure spaces indexed by $I$. Then there exists a $\sigma$-finite \emph{direct sum} measure space $\bigoplus_{i \in I} (X_i, \Sigma_{X_i}, \mu_{X_i}) \in Ob(\mathbb{M})$.
\end{proposition}
\begin{proof}
The direct sum $\bigoplus_{i \in I} (X_i, \Sigma_{X_i}, \mu_{X_i}) = \left(X, \Sigma_X, \mu_X \right) \in Ob(\mathbb{M})$ of the indexed family consists of the set $X = \bigcup_{i \in I} \left(X_i \times \{i\}\right)$, the $\sigma$-algebra $\Sigma_X = \left\{ \sigma_X : \sigma_X \subseteq X, \forall i\in I, \{x: (x, i) \in \sigma_X\} \in \Sigma_{X_i} \right\}$, and the sum measure $\mu_X(\sigma_X) = \sum_{i \in I} \mu_{X_i}(\{x: (x, i) \in \sigma_X\})$.
\end{proof}
The reader can check that the direct sum of measure spaces does not form a categorical coproduct: the copairing required to witness the universal property will not be measure-preserving.

The above propositions characterized the algebra of $\sigma$-finite measure spaces, which thus now requires base cases. Restricting our attention to the standard Borel spaces, we can take the singleton set $(I, \mathcal{B}(I), \mu_{\#})$ equipped with the counting measure $\mu_{\#}$ and the real line $(\mathbb{R}, \mathcal{B}(\mathbb{R}), \lambda)$ with the Lebesgue measure as those base cases. An $n$-fold or countable direct sum of the singleton set gives finite and countable discrete measure spaces, whose counting measure is $\sigma$-finite, while an $n$-fold product of the real line gives the Euclidean spaces, whose $n$-dimensional Lebesgue measures are $\sigma$-finite. Definition~\ref{def:borel_measure_space} will therefore formally give the class of measure spaces suitable for forming probability densities.
\begin{definition}[$\sigma$-finite standard Borel measure space]
\label{def:borel_measure_space}
The subcategory $\mathbb{M}_{\mathcal{B}} \subset \mathbb{M}$ restricts the category of measure spaces to the $\sigma$-finite \emph{standard Borel} measure spaces freely generated by finite direct products $\otimes$ (Proposition~\ref{prop:measure_space_product}) and countable direct sums $\oplus$ (Proposition~\ref{prop:measure_space_sum}) of the counting-measured singleton space $(\mathbbm{1}, \mathcal{B}(\mathbbm{1}), \mu_{\#})$ and the Lebesgue-measured reals $(\mathbb{R}, \mathcal{B}(\mathbb{R}), \lambda)$.
\end{definition}
Definition~\ref{def:borel_measure_space} covers the most common sample spaces and their base measures, as instances of a more general construction assigning base measures to finite-dimensional manifolds as sample spaces for probability densities~\cite{Radul2021}. The above only allows finite products, since the product-of-Lebesgues measure on the Hilbert cube $\mathbb{R}^{\mathbb{N}}$ (via the Borel isomorphism $\mathbb{R} \simeq [0, 1]$) fails to be $\sigma$-finite~\cite{Baker1991}. The rest of the paper will therefore work with measure spaces $\mathbb{M}_{\mathcal{B}}$, whose isomorphisms preserve base measures.

Having a class of measure spaces suitable for stating probability densities with respect to count, length, area, volume, etc., Definition~\ref{def:density_kernel} gives the class of Markov kernels which will admit densities.
\begin{definition}[Density kernel]
\label{def:density_kernel}
A standard Borel \emph{density kernel} is a $\sigma$-finite (Definition~\ref{def:sigma_finite}) Markov kernel $f: Z \leadsto X$ whose codomain forms a $\sigma$-finite standard Borel measure space $(X, \Sigma_X, \mu_{X}) \in Ob(\mathbb{M}_{\mathcal{B}})$ and which is absolutely continuous $\forall z, f(z) \ll \mu_X$ with respect to the base measure $\mu_X$
\begin{align*}
    \Dens((Z, \Sigma_Z), (X, \Sigma_X)) &:= \left\{ (f, \mu_X): (Z \leadsto X) \times \Measure{X} \mid (X, \Sigma_X, \mu_{X}) \in Ob(\mathbb{M}_{\mathcal{B}}), \forall z\in Z, f(z) \ll \mu_x \right\}.
\end{align*}
\end{definition}
Probability (and arbitrary measure) densities $\prior{}{x}{z}$ also admit an alternative interpretation as measure kernels $Z \times X \times \Sigma_{I} \rightarrow \extWeight$ whose integration under the base measure yields the normalizing constant. Proposition~\ref{thm:density_kernel} verifies that density kernels in fact admit probability densities.
\begin{theorem}[Density kernels admit densities]
\label{thm:density_kernel}
Every density kernel $(f, \mu_X): (Z \leadsto X) \times \Measure{X}$ (Definition~\ref{def:density_kernel}) into a standard Borel measure space admits a density with respect to the base measure $\mu_{X}$.
\end{theorem}
\begin{proof}
$\sigma$-finiteness of the kernel $f$ and the base measure $\mu_{X}$, plus absolute continuity, give the necessary conditions for the classical Radon-Nikodym theorem: a Radon-Nikodym derivative therefore exists
\begin{align*}
    \frac{df(z)}{d\mu_X} &: \Meas(X, \weight) &
    f(z)(\sigma_X) &= \expectint{\mu_X(dx)}{x \in \sigma_X}{\frac{df(z)}{d\mu_X}\left(x\right)}.
\end{align*}
The Radon-Nikodym derivative is the measure-theoretic notion of a probability density function
\begin{align*}
    \frac{df}{d\mu_X} &: \Meas(Z \times X, \weight), &
    \prior{f}{\cdot}{\cdot} &: \Meas(X \times Z, \weight), &
    \prior{f}{x}{z} &:= \frac{df(z)}{d\mu_X}\left(x\right).
\end{align*}
The conditions on density kernels are therefore sufficient to yield probability densities.
\end{proof}
Despite the hom-set notation used for convenience, density kernels do not form a category: identity Markov kernels are Dirac delta measures that only admit densities in discrete spaces. They do, however, support all compositional structure under which the resulting base measure still indexes a standard Borel measure space. Definition~\ref{def:density_precomp} lays the foundation for this structure.
\begin{definition}[Precomposition of a density kernel]
\label{def:density_precomp}
Given a density kernel $(f, \mu_X): (Z \leadsto X) \times \Measure{X}$ and a Markov kernel $h: W \leadsto Z$, their \emph{precomposition} is $(f, \mu_X) \circ_{\Dens} h = (f \circ_{\Stoch} h, \mu_X)$.
\end{definition}
The above precomposition gives a definition for the composition of two density kernels: given $(f, \mu_X)$ and $(g, \mu_Y)$ their composite will just be $(g \circ f, \mu_Y)$. The existence of precomposition supports a product and coproduct algebra of density kernels, as expected based on the probability algebra itself.
\begin{theorem}[Density kernels admit products and coproducts]
\label{thm:density_kernels_cartesian}
Density kernels have products $(f, \mu_X) \otimes (g, \mu_Y)$ and coproducts $(f, \mu_X) \oplus (g, \mu_Y)$, witnessed by a pairing and copairing.
\end{theorem}
\begin{proof}
Any two density kernels $(f, \mu_X): \Dens((Z, \Sigma_Z), (X, \Sigma_X))$ and $(g, \mu_Y): \Dens((Z, \Sigma_Z), (Y, \Sigma_Y))$ admit a pairing via precomposition with copying and the product measure space $(X, \Sigma_X, \mu_X) \otimes (Y, \Sigma_Y, \mu_Y) = (X \times Y, \Sigma_X \times \Sigma_Y, \mu_X \otimes \mu_Y) \in Ob(\mathbb{M}_{\mathcal{B}})$
\begin{align*}
    \left( \mcopy{Z} \fatsemi (f \otimes g), \mu_x \otimes \mu_Y \right) &: \Dens((Z, \Sigma_Z), (X, \Sigma_X) \otimes (Y, \Sigma_Y)).
\end{align*}
Any two density kernels $(f, \mu_Y): \Dens((Z, \Sigma_Z), (Y, \Sigma_Y))$ and $(g, \mu_Y): \Dens((X, \Sigma_X), (Y, \Sigma_Y))$ also admit a copairing $\binom{(f, \mu_Y)}{(g, \mu_Y)} = \left(\binom{f}{g}, \mu_Y \right)$ via the copairing of their Markov kernels in $\Stoch$.
\end{proof}
The above theorems demonstrate that density kernels represent probability densities compositionally. However, density kernels do not admit post-composition with arbitrary Markov kernels.  Section~\ref{subsec:joint_densities} will remedy this issue by applying density kernels to generate the residuals in joint Markov kernels.

\subsection{Joint densities over joint distributions}
\label{subsec:joint_densities}
Density kernels are not closed under pushforwards, and they do not form a category. $\Joint$ cannot apply directly to them. Definition~\ref{def:joint_density_kernel} therefore gives an appropriate definition for joint density kernels.
\begin{definition}[Joint density kernel]
\label{def:joint_density_kernel}
A \emph{joint density kernel} between objects $Z, X \in Ob(\Stoch)$ is a pair of a density kernel into $(M, \Sigma_M, \mu_M) \in Ob(\mathbb{M}_{\mathcal{B}})$ with a deterministic map parameterized by the residual
\begin{align*}
    \dJoint(Z, X) &:= \left\{ ((f, \mu_M), [M, k]): \Dens(Z, M) \times \Para{\Stoch_{det}}(Z, X) \mid (M, \Sigma_M, \mu_M) \in Ob(\mathbb{M}_{\mathcal{B}}) \right\}.
\end{align*}
\end{definition}
Hom-set notation once again implies these kernels form a category, which in fact they will. First, Corollary~\ref{cor:density_cond_prod} shows density kernels are closed under the joint distribution construction of Equation~\ref{eq:joint_comp}.
\begin{corollary}[Density kernels admit joint distributions as conditional products]
\label{cor:density_cond_prod}
Given a density kernel $(f_1, \mu_{M_1}): \Dens(Z, M_1)$, a measurable map $k_1: \Stoch_{det}(Z \otimes M_1, X)$, and a density kernel $(f_2, \mu_{M_2}): \Dens(X, M_2)$, composing them according to the diagram in Equation~\ref{eq:joint_comp} forms a joint density kernel
\begin{align*}
    \left( \mcopy{Z} \fatsemi ((\mcopy{M_1} \circ f_1) \otimes id_Z) \fatsemi (id_{M_1} \otimes (f_2 \circ k_1)), \mu_{M_1} \otimes \mu_{M_2}\right) &: \Dens(Z, M_1 \otimes M_2).
\end{align*}
\end{corollary}
Theorem~\ref{thm:joint_density_subcat} will show that joint density kernels form a category, and characterize them as joint Markov kernels with the extra data of a base measure on the residual.

\begin{theorem}[Joint density kernels form a category]
\label{thm:joint_density_subcat}
Joint density kernels $\dJoint$ form a wide subcategory of the restriction $\Joint_{\BorelStoch}$ of $\Joint$ to standard Borel Markov kernels in $\BorelStoch$.
\end{theorem}
\begin{proof}
First we show the joint density kernels form a subcategory, then show that subcategory is wide.

Corollary~\ref{cor:density_cond_prod} shows that density kernels are closed under the composition of $\Joint$ (Equation~\ref{eq:joint_comp}), and so along with the obvious identity morphisms and associativity law they form a category. Theorem~\ref{thm:density_kernels_cartesian} shows that this category inherits the product and coproduct structure of $\Joint$. The structure morphisms in $\Joint$ all have the unit $I$ for their residual, which admits a trivial density as a finite standard Borel space; $\dJoint$ therefore inherits the copy/delete structure of $\Joint$. This implies $\dJoint \subset \Joint_{\BorelStoch}$.

Objects and structure morphisms are inherited from $\Joint$, so the subcategory is wide.
\end{proof}
The theorem above gives a copy/delete categorical structure for joint density kernels, whose base and probability measures will be $\sigma$-finite (Definition~\ref{def:sigma_finite}) as conditions for Radon-Nikodym. There is then a precise class of measures formed by pushing forward a $\sigma$-finite measure~\cite{VakarOng2018}: the $s$-finite measures (Definition~\ref{def:s_finite}). Proposition~\ref{prop:sfkrn_cdcat} shows that such $s$-finite measure kernels form a copy/delete category.
\begin{proposition}[$s$-finite measure kernels form a CD-category~\cite{ChoJacobs2019}\footnote{Example 7.2}]
\label{prop:sfkrn_cdcat}
$s$-finite measure kernels (Definition~\ref{def:s_finite}) between measurable spaces form a CD-category $\sfKrn$ with $Ob(\sfKrn) = Ob(\Meas)$ and hom-sets given by $\sfKrn((Z, \Sigma_Z), (X, \Sigma_X)) = \left\{ f: Z \times \Sigma_X \rightarrow \extWeight \mid \forall z, f(z)\: \text{is $s$-finite} \right\}$.
\end{proposition}
$\sfKrn$ only forms a copy/delete category, not a Markov category, since different measure kernels may have different normalizing constants, including an infinite normalizing constant. Corollary~\ref{cor:sfstoch_markov} shows that restricting to probability kernels forms a Markov category.
\begin{corollary}[$s$-finite probability kernels form a Markov category]
\label{cor:sfstoch_markov}
The $s$-finite \emph{probability} kernels $f: \sfKrn((Z, \Sigma_Z), (X, \Sigma_X))$, for which $\forall z\in Z, f(z, X) = 1$, form a Markov category $\mathbf{sfStoch} \subset \mathbf{Stoch}$.
\end{corollary}
\begin{proof}
The restriction of all kernels to normalize to measure 1 renders every map $\mdel{Z}$ unique, making $I$ a terminal object and the resulting subcategory $\sfStoch$ a Markov category.
\end{proof}

Having a categorical setting capturing the Markov kernels used in computable applications, the remainder of this paper will interpret morphisms in $\dJoint$ into $s$-finite Markov kernels $\sfStoch(Z, X)$ with densities $\sfKrn(Z \otimes X, I)$. Theorem~\ref{thm:joint_density_sfkrn} shows that the joint Markov kernels of $\dJoint$ are $s$-finite and admit densities jointly measurable in the parameter and the residual.
\begin{theorem}[Joint density kernels give $s$-finite probability kernels and densities]
\label{thm:joint_density_sfkrn}
Joint density kernels $(f, [M, k]): \dJoint((Z, \Sigma_Z), (X, \Sigma_X))$ admit probability kernels $p: \sfStoch((Z, \Sigma_Z), (X, \Sigma_X))$ marginalizing out their randomness and probability densities $\prior{f}{\cdot}{\cdot}: \sfKrn((Z, \Sigma_Z) \otimes (M, \Sigma_M), I)$.
\end{theorem}
\begin{proof}
Any density kernel $f: \Dens((Z, \Sigma_Z), (M, \Sigma_M))$ gives a $\sigma$-finite probability measure and any $(M, k): \Para{\Stoch_{det}}(Z, X)$ pushes it forward. Every pushforward of a $\sigma$-finite Markov kernel is $s$-finite (Proposition~\ref{prop:sfkrn_pushforward}), so $\dJoint$ consists entirely of $s$-finite joint Markov kernels. Being $s$-finite, joint density kernels admit the required probability kernels $p: \sfStoch((Z, \Sigma_Z), (X, \Sigma_X))$ with $p(z, \sigma_X) = f(z, k(z)^{-1}(\sigma_X))$ and densities $\prior{f}{\cdot}{z}: M \times \Sigma_I \rightarrow \extWeight$ measurable in $z$ and $m$. Proposition~\ref{thm:density_kernel} defines these as the Radon-Nikodym derivative $\prior{f}{m}{z}(\{*\}) = \frac{df(z)}{d\mu_{M}}(m)$.
\end{proof}
Theorem~\ref{thm:joint_density_subcat} and Theorem~\ref{thm:joint_density_sfkrn} finally gives a desirable categorical setting: one which supports composition, products, and coproducts as a copy/delete category should, while decomposing into a deterministic causal mechanism applied to a random variable with a joint density as a structural causal model should. Section~\ref{sec:diagram_scm} will put together the machinery in this section with existing work on factorizing string diagrams syntactically to interpret those factorizations as generalizing directed graphical models. 

\section{Diagrams as causal factorizations of joint distributions and densities}
\label{sec:diagram_scm}

This section demonstrates that string diagrams with factorized densities support the full ``ladder of causation''~\cite{Pearl2018} as probabilistic models: factorized distributions, interventions, and counterfactual queries. Section~\ref{sec:gadgets} presented the $\dJoint$ construction for building up joint densities while still expressing arbitrary pushforward measures over them. Reasoning about directed graphical models or probabilistic programs compositionally requires providing a graphical syntax interpretable into $\dJoint$. Recent work~\cite{FritzKlingler2023,FritzLiang2023} treated a combinatorial syntax of string diagrams as generalized causal models. This section first reviews the definitions of a generalized causal model and its factorization of a Markov kernel, then applies that syntax to this paper's novel constructions. Doing so will enable show that via generalized causal models, joint density kernels admit factorization of their densities (Theorem~\ref{thm:factorization_joint_densities}), interventional distributions (Theorem~\ref{thm:interventional_dists}), and counterfactual distributions (Theorem~\ref{thm:counterfactuals}).


\emph{Generalized} causal models~\cite{FritzKlingler2023} provide several advantages over causal Bayesian networks as a representation of causal structure in probability models. They allow for global inputs to and outputs from a causal model, making explicit the interface necessary to reason compositionally about causal structures. It also makes explicit the grouping of ``nodes'' (in the underlying graph or hypergraph) into Markov kernels, clarifying how the joint distribution decomposes into random variables and causal mechanisms.

Definition~\ref{def:generalized_causal_model} will now describe a generalized causal model.
\begin{definition}[Generalized causal model~\cite{FritzLiang2023}]
\label{def:generalized_causal_model}
A \emph{generalized causal model} $\varphi$ over $\Sigma \in \FinHyp$\footnote{see Appendix~\ref{sec:free_constructions}} is a string diagram $p \rightarrow \dom{\tau} \leftarrow q: \FreeMarkov{\Sigma}(n, m)$ for $n, m \in \mathbb{N}$ with a bijection $q$ on wires.
\end{definition}
Any generalized causal model $p \rightarrow \dom{\tau} \leftarrow q$ is equivalent to a morphism~\cite{FritzKlingler2023}
\begin{align*}
    \varphi &: \FreeMarkov{\Sigma}\left(\bigotimes_{i=1}^{n} \tau(p(i)), \bigotimes_{j=1}^{m} \tau(q(j))\right).
\end{align*}
Definition~\ref{def:factorization} will capture factorization of a Markov kernel by a generalized causal model; Fritz and Klinger~\cite{FritzKlingler2023} called it causal compatibility in their Definition 11.
\begin{definition}[Factorization of a Markov kernel by a causal model~\cite{FritzKlingler2023}]
\label{def:factorization}
A \emph{factorization} $(f, \varphi, F)$ in $\Stoch$ consists of a morphism with decomposed domain and codomain $f: \Stoch\left(\bigotimes_{i=1}^{n} D_i, \bigotimes_{j=1}^{m} C_j\right)$,
a causal model $\varphi: \FreeMarkov{\Sigma}(n, m)$, and a strict Markov functor $F: \FreeMarkov{\Sigma} \rightarrow \Stoch$  such that $f = F(\varphi)$, $\forall i\in[1..n], D_i = F(\dom{\varphi}_i)$, and $\forall j\in[1..m], C_j = F(\cod{\varphi}_j)$.
\end{definition}


The joint density kernels $\dJoint(Z, X)$ have an important difference from the simple Markov kernels factorized by generalized causal models in Definition~\ref{def:factorization}: the density to factorize is not over $x \in X$ but over the extra structure of the residual $m \in M$. This subsection will show how to add this extra structure to a factorization, then show how to access that structure to show that generalized causal models over joint density kernels support causal inference as such: interventions and counterfactual reasoning.

Definition~\ref{def:joint_factorization} will require a factorization to label each box's residual to apply to joint Markov kernels.
\begin{definition}[Joint factorization functor]
\label{def:joint_factorization}
A \emph{joint} factorization functor for a signature $\Sigma \in \FinHyp$ is a labeling of boxes with residual wires $r: B(\Sigma) \rightarrow W(\Sigma)^{*}$ and a strict Markov functor $F: \FreeMarkov{\Sigma} \rightarrow \Joint$ respecting $\forall b \in B(\Sigma), F(b) = ([\bigotimes_{w \in r(b)} F(w), k], f): \Joint(F(\dom{b}), F(\cod{b}))$.
\end{definition}
Joint factorizations label residuals in the signature and also map to joint density kernels. Theorem~\ref{thm:factorization_joint_densities} shows they factorize the implied joint density of a causal model.
\begin{theorem}[Joint density kernels admit factorized densities]
\label{thm:factorization_joint_densities}
Given a signature $\Sigma \in \FinHyp$, a strict Markov functor $F: \FreeMarkov{\Sigma} \rightarrow \dJoint$ gives a joint density $\prior{f}{\cdot}{\cdot \in F(\dom{\varphi})}$ for every causal model $\varphi: \FreeMarkov{\Sigma}(n, m)$.
\end{theorem}
\begin{proof}
Definition~\ref{def:joint_factorization} requires for any sub-diagram $\varphi' \subseteq \varphi$ there will be some $F(\varphi') = (f, [M, k])$. Theorem~\ref{thm:joint_density_sfkrn} then gives a density over the residual, while the functoriality of $F$ and Corollary~\ref{cor:density_cond_prod} together imply that products of individual joint-densities yield the complete joint density.
\end{proof}
Theorem~\ref{thm:interventional_dists} then shows that by assigning boxes optional points in their codomains, joint factorizations also admit interventional distributions.
\begin{theorem}[Joint factorizations admit interventional distributions]
\label{thm:interventional_dists}
Consider a joint factorization $(f, \varphi, F)$ over a signature $\Sigma$. Then any intervention $\mathbf{do}: \prod_{b: B(\Sigma)} I \oplus \C_{det}(I, F(\cod{b})$ induces a functor $\mathrm{Int}: \FreeMarkov{\Sigma} \rightarrow \Joint$ and an interventional distribution $\mathrm{Int}(\varphi)$.
\end{theorem}
\begin{proof}
Any single-box free string diagram has an image $F(\left<b\right>)$. We define the required functor $\mathrm{Int}: \FreeMarkov{\Sigma} \rightarrow \Joint$ by extension of a hypergraph morphism $\alpha: \Sigma \rightarrow \hyp{\Joint}$ following Fritz and Liang~\cite{FritzLiang2023} (see their Remark 4.3). $\alpha$ will be identity on wires and intervene on boxes
\begin{align*}
    \alpha(b) &: B(\Sigma) \rightarrow B(\hyp{\Joint}) \\
    \alpha(b) &= \begin{cases}
        \hyp{([I, \mdel{\dom{b}}], \mdel{\dom{b}} \fatsemi x)} & \mathbf{do}(b) = \mathbf{inr}(x) \\
        \hyp{F(\left<b\right>)} & \mathbf{do}(b) = \mathbf{inl}(I)
    \end{cases}. \qedhere
\end{align*}
\end{proof}

Finally, Theorem~\ref{thm:counterfactuals} employs similar reasoning to model counterfactual queries over jointly factorized causal models, given fixed values for random variables and an intervention.
\begin{theorem}[Joint factorizations give counterfactuals]
\label{thm:counterfactuals}
Consider a signature $\Sigma \in \FinHyp$ and a joint factorization $(f, \varphi, F)$. Then any intervention $\mathbf{do}: \prod_{b: B(\Sigma)} I \oplus \C_{det}(I, F(\cod{b})$ and any assignment $U: B(\Sigma) \rightarrow [0, 1]$ of uniform random variates to boxes induces a functor $\mathrm{If}: \FreeMarkov{\Sigma} \rightarrow \Joint$ and a counterfactual distribution $\mathrm{If}(\varphi)$.
\end{theorem}
\begin{proof}
We work as above, but this time explicitly consider the structure of the image $F(\left<b\right>)=(f, [M, k])$. $f$ gives a standard Borel probability measure, so the Randomization Lemma~\cite{Bogachev2007} demonstrates equality of $f$ with a pushforward $F(\left<b\right>)_1 = f(\cdot) = g(\cdot, z)_{*}(U)(du)$ of the uniform distribution $U(du)$ by a deterministic map $g(\cdot, z)$. Our hypergraph morphism utilizes that fact
\begin{align*}
    \alpha(b) &= \begin{cases}
        \hyp{([I, \mdel{\dom{b}}], \mdel{\dom{b}} \fatsemi x)} & \mathbf{do}(b) = \mathbf{inr}(x) \\
        \hyp{ \dirac{U(b)}{g(b, \cdot)} } & \mathbf{do}(b) = \mathbf{inl}(I)
    \end{cases}. \qedhere
\end{align*}
\end{proof}

Together, Theorems \ref{thm:factorization_joint_densities}, \ref{thm:interventional_dists}, and \ref{thm:counterfactuals} demonstrate that joint density kernels, jointly factorized by a generalized causal model, support the properties that have made directed graphical models so widely useful. With these theorems as ``sanity checks'', Section~\ref{sec:discussion} will summarize the paper's overall contributions, give some worked examples applying $\dJoint$, and discuss future work.

\section{Discussion}
\label{sec:discussion}

This paper started from the existing work on copy/delete categories, Markov categories, and the factorization of morphisms in those categories by generalized causal models. From there, Section~\ref{sec:gadgets} constructed a novel Markov category $\Joint$ whose morphisms keep internal track of the joint distribution they denote, defined a subcategory $\dJoint \subset \Joint$ whose morphisms support only joint densities over standard Borel spaces as their internal distributions. Section~\ref{sec:diagram_scm} then demonstrated that $\Joint$ supports factorization by generalized causal models, that these factorize joint densities $\dJoint$, and that these support the interventional and counterfactual reasoning necessary for causal inference. This section will discuss some short worked examples of using $\dJoint$ for real probability models (Section~\ref{subsec:examples}), and then move on to speculate what future work could spring from the paper's developments (Section~\ref{subsec:conclusion}).

\subsection{Worked examples}
\label{subsec:examples}

\begin{figure}
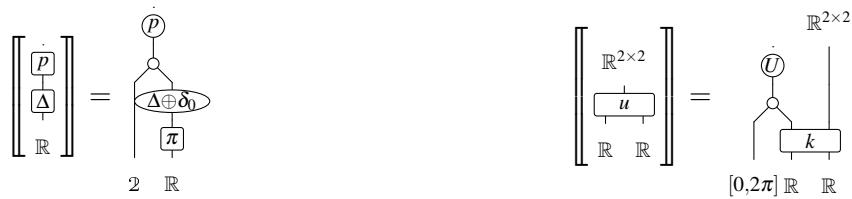

    \centering
    \begin{subfigure}[t!]{0.45\textwidth}
        \begin{align*}
            \sem{\usesatexfig{fakeCoinGraph}} &= \usesatexfig{fakeCoinDiagram}
        \end{align*}
        \caption{The wiring diagram in $\dJoint$ of a mixture model between a delta and a Gaussian, and its image in $\Stoch$ with a coproduct projection}
        \label{subfig:fake_coin}
    \end{subfigure}
    \qquad
    \begin{subfigure}[t!]{0.45\textwidth}
        \begin{align*}
            \sem{\usesatexfig{uniformCircleGraph}} &= \usesatexfig{uniformCircleDiagram}
        \end{align*}
        \caption{A Markov kernel in $\dJoint$ projecting a sample from the uniform circle through a linear transformation, and its image in $\Stoch$}
        \label{subfig:uniform_circle}
    \end{subfigure}
    \caption{Example joint density kernels (Definition~\ref{def:joint_density_kernel}): a mixture model between a constant and a Gaussian distribution depending on a coin flip (\textit{left}) and a Markov kernel projecting a random angle onto a parametrically skewed ellipse (\textit{right}). The $\sem{\cdot}$ functor (Corollary~\ref{cor:joint_markov}) maps into $\Stoch$.}
    \label{fig:example_diagrams}
\end{figure}

The previous sections have focused on formalism. Section~\ref{sec:gadgets} defined a Markov category $\dJoint$ of joint density kernels in $\Stoch$ (rather than the typical restriction to $\mathbf{FinStoch}$) whose residuals (by construction) admit probability densities. Section~\ref{sec:diagram_scm} then established that the generalized causal models recently described in the categorical probability literature can indeed apply to $\dJoint$ morphisms, factorizing their joint densities and providing for causal reasoning. This subsection will apply the $\dJoint$ formalism to the models shown in Figure~\ref{fig:example_diagrams}, taken from Wu et al~\cite{Wu2018} and Radul and Alexeev~\cite{Radul2021}.

Figure~\ref{subfig:fake_coin} shows a generative model in which we detect fake coins by placing an  even number of coins on a well-calibrated balance. The presence of a fake coin, whose weight deviates from the others, will tip the balance away from the neutral position. $p$ determines whether the a fake coin is present, which in turn determines whether the balance position is distributed according to a Gaussian $\Delta \sim \mathcal{N}(1, 0.5)$ or according to a Dirac measure $\Delta \sim \delta_{0}$. The joint distribution shown on the right-hand side of the equation admits a density with respect to the standard Borel measure space $(\mathbbm{2}, \mathcal{B}(\mathbbm{2}), \mu_{\#}) \otimes ((\mathbb{R}, \mathcal{B}(\mathbb{R}), \lambda) \oplus (\mathbbm{1}, \mathcal{B}(\mathbbm{1}), \mu_{\#}))$, whereas the marginal on $\mathbb{R}$ lacks a density for the Lebesgue measure $\lambda$.

Figure~\ref{subfig:uniform_circle} shows the example from Radul and Alexeev~\cite{Radul2021} in which a sample from $U(0, 2\pi)$ is projected onto a non-isotropic ellipse. Those authors calculate a probability density on the ellipse via the projection's Jacobian. Figure~\ref{subfig:uniform_circle} shows the two components of a $\dJoint$ morphism: how the uniformly random angle $U$ and a linear transformation $\mathbb{R}^{2\times 2}$ parameterize the the geometric projection $k$. The equation shows how $\sem{\cdot}$ maps the single box in $\dJoint$ (left) to the Markov kernel in $\Stoch$ (right).

The two examples in Figure~\ref{fig:example_diagrams} both show how the $\dJoint$ construction can compactly encode complex, parameterized joint probability densities linked by deterministic causal mechanisms. Section~\ref{subsec:conclusion} will discuss potential future work extending this paper's construction and conclude.

\subsection{Future work and conclusion}
\label{subsec:conclusion}
This paper's mathematical constructions could generalize or be strengthened in a number of ways. It would be desirable to obtain a category in which Markov kernels admit common-sense densities without having to separate into a density over a standard Borel space and a pushforward through a deterministic map; the Lebesgue decomposition of arbitrary measures into mutually singular absolutely-continuous, diffuse, and atomic portions suggests a possible route to that goal. Up to a normalization constant, every reference measure in $\mathbb{M}$ is a Hausdorff measure. This suggests densities could be obtained by considering manifolds, standardizing on the Hausdorff measure as Radul and Alexeev~\cite{Radul2021} suggest, and then defining density kernels on that foundation. Finally, Definition~\ref{def:borel_measure_space} forms an endofunctor in the category of measure spaces whose algebras and coalgebras may prove of interest. For example, recent work by Dash~\cite{Dash2023} explored defining probability measures on quasi-Borel spaces as pushforwards of a uniform distribution on the Hilbert cube, an element of the endofunctor's terminal coalgebra.

Future work can go in a number of directions to unify the formalisms of applied probabilistic reasoning. Instantiating this paper's constructions in a Markov category in which all randomness arises from an independent noise source would transform any causal factorization of a joint (density) kernel into a structural causal model~\cite{Pearl2012}, unifying causal Bayes nets with structural equation models. In the application area of probabilistic programming, this paper has only described ``first-order'' probabilistic programming languages lacking general stochastic recursion~\cite{vanDeMeent2018}, corresponding to non-closed Markov categories. A combinatorial syntax for hierarchical string diagrams~\cite{AlvarezPicallo2022} would extend our reasoning in this paper to the closed Markov categories such as $\QBS$~\cite{Heunen2017} that provide denotations for higher-order probabilistic programming languages.  We intend to extend this paper's formalism to categorify Sequential Monte Carlo methods~\cite{Naesseth2019} for generalized causal models of unnormalized distributions. We aim to apply the $\dJoint$ construction alongside recent work on unique name generation~\cite{Sabok2021} to model heterogeneous tracing in probabilistic programming. Recent work on free string diagrams~\cite{Wilson2023} has also suggested ways to map from free string diagrams to free diagrams of optics; equipping joint density kernels with optic structure would follow up on the work of Smithe~\cite{Smithe2020} and Schauer~\cite{Schauer2023}.

\paragraph{Acknowledgements} We would like to thank the anonymous reviewers for their feedback and advice in refining the paper for camera-ready. We would also like to thank the ACT 2023 program chairs for their careful shepherding of the review process. We thank Tobias Fritz, Luke Ong, Sam Staton, and Matthijs Vákár for laying the categorical foundations of $s$-finite Markov kernels. Finally, we would like to extensively thank Alex Lew for early discussions and cooperation on preliminary work to this paper. Eli Sennesh was supported by NSF award 2047253.

\nocite{*}
\bibliographystyle{eptcs}
\bibliography{example}

\appendix

\newpage

\vspace{-1em}
\section{Measure theory background}
\label{sec:measure_theory}
Measure theory studies ways of assigning a ``size'' to a set (beyond its cardinality); these can include count, length, volume, and probability. Definition~\ref{def:standard_borel_space} begins with a nice class of measurable spaces.
\begin{definition}[Standard Borel space]
\label{def:standard_borel_space}
Let  $(X, T_X) \in Ob(\mathbf{Top})$ be a separable complete metric space or homeomorphic to one. Equipping $X$ with its Borel $\sigma$-algebra $\mathcal{B}(X)$ generated by complements, countable unions, and countable intersections of open subsets $U \in T$ yields a \emph{standard Borel space} $(X, \mathcal{B}(X)) \in Ob(\Sbs)$, which is also a measurable space since $\Sbs \subset \Meas$.
\end{definition}
The paper uses standard Borel spaces as a basis for its category of measure spaces (Definition~\ref{def:borel_measure_space}). Example~\ref{ex:unit_interval} is such a space.
\begin{example}[The unit interval]
\label{ex:unit_interval}
The closed unit interval $[0, 1]$ with its Borel $\sigma$-algebra of open sets $\mathcal{B}(0, 1)$ forms a standard Borel space $([0, 1], \mathcal{B}(0, 1))$.
\end{example}
Having a category of measurable spaces and some nice examples, Definition~\ref{def:measure} formally defines what it means to assign a ``size'' to a measurable set.
\begin{definition}[Measure]
\label{def:measure}
A \emph{measure} $\mu: \Measure{Z}$ on a measurable space $(Z, \Sigma_Z) \in Ob(\Meas)$ is a function $\mu: \Sigma_Z \rightarrow \extWeight$ that is null on the empty set ($\mu(\emptyset) = 0$) and countably additive over pairwise disjoint sets
\begin{align*}
    \inference{
        \{\sigma_k \in \Sigma_Z\}_{k\in\mathbb{N}} &
        \forall k\in\mathbb{N}, n\in\mathbb{N}, n \neq k \implies \sigma_k \cap \sigma_n = \emptyset
    }{
        \mu\left( \bigcup_{k \in \mathbb{N}} \sigma_k \right) = \sum_{k\in\mathbb{N}} \mu(\sigma_k)
    }
\end{align*}
\end{definition}
Reasoning compositionally about measure requires a class of maps between a domain and a codomain that form measures. The Giry monad~\cite{Giry1982} sends a measurable space $(X, \Sigma_X)$ to its space of measures $\Measure{X}$ and probability measures $\Prob{X} \subset \Measure{X}$.
Definition~\ref{def:measure_kernel} defines maps into those spaces, treating the domain as a parameter space for a measure over the codomain.
\begin{definition}[Measure kernel]
\label{def:measure_kernel}
A \emph{measure kernel} between two measurable spaces $(Z, \Sigma_Z), (X, \Sigma_X) \in Ob(\Meas)$ is a function $f: Z \times \Sigma_X \rightarrow \extWeight$ such that $\forall z\in Z, f(z, \cdot) : \Measure{X}$ is a measure and $\forall \sigma_X \in \Sigma_X, f(\cdot, \sigma_X): \Meas((Z, \Sigma_Z), (\extWeight, \Sigma_{\extWeight}))$ is measurable.
\end{definition}
Measure kernels serve both to define Markov kernels below, and to form a broader class of copy/delete categories, which in Theorem~\ref{thm:joint_density_sfkrn} are seen to admit probability densities as morphisms. Definition~\ref{def:markov_kernel} specializes to measure kernels yielding only normalized probability measures.
\begin{definition}[Markov kernel]
\label{def:markov_kernel}
A \emph{Markov kernel} is a measure kernel $f: Z \times \Sigma_X \rightarrow \extWeight$ whose measure is a probability measure so that $\forall z\in Z, f(z, \cdot): \Prob{X}$ and $\forall z\in Z, f(z, X) = 1$.
\end{definition}
The Giry monad, restricted to probability spaces, yields Markov kernels as its Kleisli morphisms $\Meas((Z, \Sigma_Z), \Measure{X})$, forming the main category of Markov kernels in this paper ($\Stoch$, Definition~\ref{def:stoch}). Describing densities categorically then requires invoking the Radon-Nikodym Theorem, which determines when probability measures have densities. The next two definitions give the Theorem's conditions, which must be satisfied for a density to exist.

Definition~\ref{def:sigma_finite} will formalize the condition that both the base measure and a probability measure consist of sums over countable partitions of the sample space.
\begin{definition}[$\sigma$-finite measure kernel]
\label{def:sigma_finite}
A \emph{$\sigma$-finite measure kernel} $f: Z \times \Sigma_X \rightarrow \extWeight$ is a measure kernel which at every parameter $z \in Z$ splits its codomain into countably many measurable sets $X = \bigcup_{n \in \mathbb{N}} X_n \in \Sigma_X$, each of which has finite measure $f(z)(X_n) < \infty$.
\end{definition}
Definition~\ref{def:absolute_continuity} will now formalize the further requirement that for a probability measure to admit a density function, it must have only the same null-sets as the underlying base measure.
\begin{definition}[Absolute continuity]
\label{def:absolute_continuity}
One $\sigma$-finite measure kernel $f: Z \times \Sigma_X \rightarrow [0, \infty]$ is \emph{absolutely continuous} ($f \ll g$) with respect to another $\sigma$-finite measure kernel over the same codomain $g: Y \times \Sigma_X \rightarrow [0, \infty]$ when $\forall z\in Z, y\in Y, \sigma_X \in \Sigma_X, g(y)(\sigma_X) = 0 \implies f(z)(\sigma_X) = 0$.
\end{definition}
The conditions in Definition~\ref{def:sigma_finite} and Definition~\ref{def:absolute_continuity} are necessary and sufficient for the existence of a probability density via the Radon-Nikodym Theorem, as used in density kernels in Definition~\ref{def:density_kernel}. Density kernels use measure \emph{spaces} as their codomains: these group together the desired topology, dimensionality, and base measure. Definition~\ref{def:measure_space} below formally defines measure spaces, which the paper uses in the specific form of standard Borel measure spaces (Definition~\ref{def:borel_measure_space}).
\begin{definition}[Measure space]
\label{def:measure_space}
A \emph{measure space} is a pair $((X, \Sigma_X), \mu)$ of a measurable space $(X, \Sigma_X) \in Ob(\Meas)$ with a measure $\mu: \Measure{X}$ on that space.
\end{definition}
The measure spaces just defined form objects in a category which Definition~\ref{def:measure_space_cat} describes. Passing from the category of measurable spaces $\Meas$ to the category of measure spaces $\M$ requires the resulting morphisms to respect the chosen measure, so that measurable sets do not ``grow'' or ``shrink''.

Having given the conditions for densities to exist, the paper passes from density kernels to joint density kernels. Definition~\ref{def:s_finite} will give a class of Markov kernels encompassing all those in this paper, particularly joint density kernels.
\begin{definition}[$s$-finite measure kernel]
\label{def:s_finite}
An \emph{$s$-finite measure kernel} $f: Z \times \Sigma_X \rightarrow \extWeight$ is a measure kernel (as in Definition~\ref{def:measure_kernel} above) which decomposes into a sum of finite kernels $f = \sum_{n \in \mathbb{N}} f_n$ such that $\forall n \in \mathbb{N}, f_n: Z \times \Sigma_X \rightarrow \extWeight$ and $\forall n \in \mathbb{N}, \exists r_n \in \weight, \forall z \in Z, f_n(z, X) \leq r_n$.
\end{definition}
Proposition~\ref{prop:sfkrn_pushforward} will demonstrate that the class of $s$-finite kernels (Definition~\ref{def:s_finite}) includes all pushforwards of $\sigma$-finite kernels, and therefore the pushforwards of all measure kernels admitting densities.
\begin{proposition}[$s$-finite kernels are pushforwards of $\sigma$-finite kernels~\cite{VakarOng2018,Staton2017}]
\label{prop:sfkrn_pushforward}
A measure kernel $f: Z \times \Sigma_X \rightarrow \extWeight$ is $s$-finite if and only if it is a pushforward $f = \mcopy{Z}\fatsemi (p \otimes id_{Z}) \fatsemi k$ of a $\sigma$-finite measure kernel $p$ through a deterministic $k$.
\end{proposition}
The above proposition includes trivial pushforwards, so every $\sigma$-finite (Definition~\ref{def:sigma_finite}) measure kernel is $s$-finite (Definition~\ref{def:s_finite}) but not the other way around.

\section{Parametric and coparametric categories}
\label{sec:para_copara}
This section will review the definitions of parametric and coparametric (bi)categories, first given in the categorical cybernetics literature~\cite{Capucci2021}. For the sake of rigor, the reader can also see a recent review on actegories~\cite{Capucci2022}. As a starting point, Definition~\ref{def:M_actegory} will describe how a symmetric monoidal category (SMC) can ``act upon'' another category functorially.

\begin{definition}[$\M$-actegory]
\label{def:M_actegory}
Consider a symmetric monoidal category $(\M, J, \odot)$ and a category $\C$. An $\M$-actegory is a pair of the two with a functor $\bullet: \M \times \C \rightarrow \C$ from the product category and natural transformations $\varepsilon: J \bullet X \simeq X$ and $\delta: (M \bullet N)\bullet X = M \bullet (N \bullet X)$.
\end{definition}

Definition~\ref{def:para_full} will then apply the actegory concept to define a bicategory whose morphisms accumulate parameters in the course of composition.
\begin{definition}[Parametric categories~\cite{Capucci2021}]
\label{def:para_full}
Given an $\M$-actegory $\C$, the \emph{parametric (bi)category} \(\ParaB{\C}\) has as objects those of $\C$ and as morphisms the pairs \(\ParaB{\C}(A, B) = \left\{ (M, k) \in Ob(\M) \times \C(M \bullet A, B) \right\}\). Composition for morphisms \((M, k): \ParaB{\C}(A, B)\) and \((M', k'): \ParaB{\C}(B, C)\) consists of \((M' \odot M, k' \circ (id_{M'} \bullet k)))\) while identities on objects \(A\) consist of \((I, id_{A})\).
\end{definition}
Parametric (bi)categories of course have a dual, definable as $\ParaB{\C^{op}}^{op}$. Definition~\ref{def:copara_full} will describe this category, whose morphisms admit ``coparameters'' accumulate extra elements of the codomain.
\begin{definition}[Coparametric categories~\cite{Capucci2021}]
\label{def:copara_full}
Given an $\M$-actegory $\C$, the \emph{coparametric category}
\begin{align*}
    \CoParaB{\C} &\in Ob(\mathbf{Cat})
\end{align*}
has as objects those of $\C$ and as morphisms \(\CoParaB{\C}(A, B)\) the pairs \((M, f) \in Ob(\M) \times \C(A, M \bullet B)\). Composition for \((M, f): \CoParaB{\C}(A, B)\) and \((M', g): \CoParaB{\C}(B, C)\) consists of \((M \odot M', (id_{M} \bullet g) \circ f))\) while identities on objects \(A\) consist of \((I, id_{A})\).
\end{definition}
The coparametric category construction generalizes the idea of a writer monad to more than one object, and represents morphisms that ``log'' or ``leave behind'' a cumulative effect. Definition~\ref{def:smc_actegory} will describe symmetric monoidality for the $\M$-actegory on $\C$ when $(\C, \otimes I)$ is symmetric monoidal.
\begin{definition}[Symmetric monoidal $\M$-actegory]
\label{def:smc_actegory}
A \emph{symmetric monoidal $\M$-actegory} is an $\M$-actegory $\C$ equipped with a symmetric monoidal structure and a natural isomorphism $\kappa_{M,X,Y}: M \bullet (X \otimes Y) \simeq X \otimes (M \bullet Y)$, satisfying coherence laws similar to those of a costrong comonad.
\end{definition}

Finally, Proposition~\ref{prop:para_copara_monoidal} will demonstrate that given a symmetric monoidal actegory as in Definition~\ref{def:smc_actegory}, the constructions above admit symmetric monoidal structure themselves.
\begin{proposition}[Parametric and coparametric categories admit monoidal structure~\cite{Capucci2022}\footnote{Example 5.1.8}]
\label{prop:para_copara_monoidal}
Given a symmetric monoidal $\M$-actegory $(\C, \otimes, I)$, the parametric bicategory $\ParaB{\C}$ and coparametric bicategory $\CoParaB{\C}$ form symmetric monoidal bicategories $(\ParaB{\C}, \otimes, I)$ and $(\CoParaB{\C}, \otimes, I)$.
\end{proposition}

\section{Free copy/delete and Markov categories}
\label{sec:free_constructions}

Generalized causal models~\cite{FritzKlingler2023} employ hypergraphs, which ``flip'' the status of nodes and edges relative to ordinary graphs: ``hypernodes'' are drawn as wires and ``hyperedges'' connecting them as boxes. These hypergraphs represent string diagrams combinatorially; restricting hypergraphs to conditions matching certain kinds of categories defines ``free'' categories of those kinds. This subsection will build up free copy/delete and Markov categories with generalized causal models as morphisms.

Definition~\ref{def:hypergraph} defines hypergraphs via sets~\cite{Gallo1993}; Bonchi et al~\cite{Bonchi2016} provides categorical intuition.
\begin{definition}[Hypergraph]
\label{def:hypergraph}
A \emph{hypergraph} is a 4-tuple $(W, B, \mathrm{dom}, \mathrm{cod})$ consisting of a set of vertices, nodes, or ``wires'' $W$; a set of hyperedges or ``boxes'' $B$; a function $\mathrm{dom}: B \rightarrow W^{*}$ assigning a domain to each box; and a function $\mathrm{cod}: B \rightarrow W^{*}$ assigning a codomain to each box.

We abuse notation and write individual boxes $b \in B: \dom{b} \rightarrow \cod{b}$.
\end{definition}
Definition~\ref{def:hyp_morphism} specifies relabelings of one hypergraph's wires and boxes with those of another.
\begin{definition}[Hypergraph morphism]
\label{def:hyp_morphism}
Given hypergraphs $G, H$, a \emph{hypergraph morphism} $\alpha: G \rightarrow H$ is a pair of functions assigning wires to wires and boxes to boxes, the latter respecting the former
\vspace{-0.5em}
\begin{align*}
    \Hyp(G, H) &:= \left\{ (\alpha_W, \alpha_B) \in W(H)^{W(G)} \times B(H)^{B(G)} \mid \forall b \in B(G), \alpha_B(b): \alpha_W(\dom{b}) \rightarrow \alpha_W(\cod{b}) \right\}.
\end{align*}
\end{definition}
As implied by the hom-set notation, hypergraphs and their morphisms form a category $\Hyp$ \cite{Bonchi2016}, and our application will employ the full subcategory $\FinHyp$ in which $W$ and $B$ both have finite cardinality. Finally, a hypergraph $H$ is \emph{discrete} when $B(H) = \emptyset$; $\underline{n}$ denotes a discrete hypergraph with $n\in\mathbb{N}$ wires. Any monoidal category has a (potentially infinite) underlying hypergraph, which we denote $\hyp{\cdot}: \mathbf{MonCat} \rightarrow \Hyp$ following Fritz and Liang~\cite{FritzLiang2023}.

Often a finite hypergraph $\Sigma \in \FinHyp$ denotes the generating objects and morphisms of a free monoidal category, or the primitive types and functions of a domain-specific programming language. We call such a finite hypergraph a \emph{monoidal signature}. Definition~\ref{def:free_cd} formally defines the copy/delete category freely generated by a signature $\Sigma$, which Definition~\ref{def:free_markov} will restrict to free Markov categories.
\begin{definition}[Free copy/delete category for the signature $\Sigma$~\cite{FritzLiang2023}]
\label{def:free_cd}
The \emph{free CD category} $\FreeCD{\Sigma}$ for $\Sigma \in \FinHyp$ is a subcategory $\FreeCD{\Sigma} \subseteq \mathbf{cospan}(\FinHyp / \Sigma)$ where
\begin{itemize}
    \item Objects are the pairs $(n, \sigma) \in \mathbb{N} \times \underline{n} \rightarrow \Sigma$ assigning outer wires of a string diagram to wires in $\Sigma$;
    \item Morphisms are isomorphism classes of cospans, given combinatorially
    \vspace{-1em}
    \begin{multline*}
        \FreeCD{\Sigma}((n, \sigma_n), (m, \sigma_m)) = \\ \left\{ p \rightarrow \dom{\tau} \leftarrow q \in \FinHyp(\underline{n}, \dom{\tau}) \times Ob(\FinHyp/\Sigma) \times \FinHyp(\underline{m}, \dom{\tau}) \right\},
    \end{multline*}
    such that $\tau: G \rightarrow \Sigma \in Ob(\FinHyp/\Sigma)$ is a hypergraph morphism from an acyclic $G$ and every wire $w \in W(G)$ has at most one ``starting place'' as the diagram's input or a box's output
    \begin{align*}
        |p^{-1}(w)| + \sum_{b \in B(G)} \sum_{w' \in \cod{b}} \mathbb{I}[w' = w] &\leq 1.
    \end{align*}
\end{itemize}
\end{definition}
Intuitively, a morphism in $\FreeCD{\Sigma}$ is syntax specifying a string diagram with no looping or merging wires, whose boxes and wires are labeled by $\Sigma$. Definition~\ref{def:free_markov} passes to the free Markov category $\FreeMarkov{\Sigma}$ just by syntactically enforcing the naturality of $\mdel{Z}$.
\begin{definition}[Free Markov category for the signature $\Sigma$]
\label{def:free_markov}
The \emph{free Markov category} $\FreeMarkov{\Sigma}$ for $\Sigma \in \FinHyp$ is the wide subcategory of $\FreeCD{\Sigma}$ restricted to morphisms in which every output from every box connects to somewhere else
\vspace{-0.5em}
\begin{align*}
    \mathrm{connects}(w, G, q) &:= \mathbb{I}[\exists b \in B(G):w\in\cod{b} \implies q^{-1}(w) \neq \emptyset \vee \exists b' \in B(G):w\in\dom{b'}]
\end{align*}
\begin{multline*}
    \FreeMarkov{\Sigma}(n, m) := \\ \left\{ p \rightarrow \dom{\tau} \leftarrow q \in \FreeCD{\Sigma}(n, m) \mid \forall w \in W(\dom{\tau}), \mathrm{connects}(w, \dom{\tau}, q) \right\},
\end{multline*}
and with composition redefined to syntactically enforce this by iterating the deletion of discarded boxes to a fixed-point after composition in $\FreeCD{\Sigma}$.
\end{definition}

\end{document}